\documentclass[lettersize,journal]{IEEEtran}
\IEEEoverridecommandlockouts

\usepackage{amssymb}
\usepackage[cmex10]{amsmath}
\usepackage{stfloats}
\usepackage{graphicx}
\usepackage{subfigure}
\usepackage{tabularx}
\usepackage{epsfig,epsf,color,balance,cite}
\usepackage{verbatim}
\usepackage{url}
\usepackage{bm}
\usepackage{booktabs}

\usepackage{amsmath}
\allowdisplaybreaks[2]
\usepackage{ntheorem}
\newtheorem{theorem}{Theorem}
\newtheorem{lemma}{Lemma}
\newtheorem*{proof}{Proof}

\usepackage{algorithm}
\usepackage{algorithmic}

\usepackage{color}
\definecolor{myc1}{rgb}{0,0,0}

\begin{document}

% paper title
\title{A Joint Communication and Computation Design for Probabilistic Semantic Communications}

\author{Zhouxiang Zhao,~\IEEEmembership{Graduate Student Member,~IEEE,}
        Zhaohui Yang,~\IEEEmembership{Member,~IEEE,}
        Mingzhe Chen,\\ \IEEEmembership{Member,~IEEE,}
%        Chongwen Huang,~\IEEEmembership{Member,~IEEE,}
%        Qianqian Yang,~\IEEEmembership{Member,~IEEE,}
%        Hongyang Chen,~\IEEEmembership{Senior Member,~IEEE,}
%        Wei Xu,~\IEEEmembership{Senior Member,~IEEE,}
        Zhaoyang Zhang,~\IEEEmembership{Senior Member,~IEEE,}
        and H. Vincent Poor,~\IEEEmembership{Life Fellow,~IEEE}
% \thanks{Manuscript received April 19, 2021; revised August 16, 2021. This work is supported by Zhejiang Lab Program under grant K2023QA0AL02, and Zhejiang Science and Technology Program under grant 2023C01021. \textit{(Corresponding author: Zhaohui Yang.)}}
%\thanks{Z. Zhao, Z. Yang, C. Huang, Q. Yang, and Z. Zhang are with College of Information Science and Electronic Engineering, Zhejiang University, and also with Zhejiang Provincial Key Laboratory of Info. Proc., Commun. \& Netw. (IPCAN), Hangzhou, China (e-mails: \{zhouxiangzhao, yang\_zhaohui, chongwenhuang, qianqianyang20, ning\_ming\}@zju.edu.cn).}
\thanks{Z. Zhao, Z. Yang, and Z. Zhang are with College of Information Science and Electronic Engineering, Zhejiang University, and also with Zhejiang Provincial Key Laboratory of Info. Proc., Commun. \& Netw. (IPCAN), Hangzhou, China (e-mails: \{zhouxiangzhao, yang\_zhaohui, ning\_ming\}@zju.edu.cn).}
\thanks{M. Chen is with Department of Electrical and Computer Engineering and Institute for Data Science and Computing, University of Miami, Coral Gables, FL, 33146, USA (e-mail: mingzhe.chen@miami.edu).}
\thanks{H. V. Poor is with the Department of Electrical and Computer Engineering, Princeton University, Princeton, NJ 08544, USA (e-mail: poor@princeton.edu).}
%\thanks{H. Chen is with Zhejiang Lab, Hangzhou 311121, China (e-mail: dr.h.chen@ieee.org).}
%\thanks{W. Xu is with National Mobile Communications Research Laboratory, Southeast University, Nanjing, China (e-mail: wxu@seu.edu.cn).}
}

% % The paper headers
% \markboth{Journal of \LaTeX\ Class Files,~Vol.~14, No.~8, August~2021}%
% {Shell \MakeLowercase{\textit{et al.}}: A Sample Article Using IEEEtran.cls for IEEE Journals}

% \IEEEpubid{0000--0000/00\$00.00~\copyright~2021 IEEE}
% % Remember, if you use this you must call \IEEEpubidadjcol in the second
% % column for its text to clear the IEEEpubid mark.

% make the title area
\maketitle

\begin{abstract}
In this paper, the problem of joint transmission and computation resource allocation for a multi-user probabilistic semantic communication (PSC) network is investigated. In the considered model, users employ semantic information extraction techniques to compress their large-sized data before transmitting them to a multi-antenna base station (BS). Our model represents large-sized data through substantial knowledge graphs, utilizing shared probability graphs between the users and the BS for efficient semantic compression. The resource allocation problem is formulated as an optimization problem with the objective of maximizing the sum of equivalent rate of all users, considering total power budget and semantic resource limit constraints. The computation load considered in the PSC network is formulated as a non-smooth piecewise function with respect to the semantic compression ratio. To tackle this non-convex non-smooth optimization challenge, a three-stage algorithm is proposed where the solutions for the receive beamforming matrix of the BS, transmit power of each user, and semantic compression ratio of each user are obtained stage by stage. Numerical results validate the effectiveness of our proposed scheme.
\end{abstract}

\IEEEpeerreviewmaketitle

\section{Introduction}
\IEEEPARstart{T}{he} rapid development of wireless communication technology has initiated an era of unprecedented connectivity \cite{10024766} that brings with it a growing complexity of data transmission. Moreover, the principles of information theory have undeniably shaped modern communication systems. While this model has been invaluable, it inherently falls short in capturing the richer semantic dimension of the information being exchanged \cite{9771334}. In response to the limitations of traditional information theory, the concept of semantic communication has emerged as a compelling technology \cite{9955525} to handle the growing complexity of data transmission. Semantic communication transcends the mere exchange of abstract symbols, instead placing an emphasis on the meaning and purpose of a message \cite{10233741}. Different from conventional communications that focuses on data rate maximization, semantic communications prioritizes data meaning transmission.
% , which can be applied to emerging intelligent applications including intelligent Internet of robotics.

The advent of semantic communication has gained significant attention in the realm of communication research, representing a departure from established paradigms \cite{chaccour2022less}. However, despite its growing importance, the concept of semantic communication remains in a state of ongoing evolution \cite{10000901} characterized by the lack of a universally accepted definition, a comprehensive theoretical framework, and a unified understanding \cite{9679803}. Research in this field is exploratory, reflecting the challenges and opportunities of semantic communication in modern communication systems.

To achieve the advantages of semantic communication, one of the intriguing challenges is how to effectively obtain key performance indicators (KPIs) for performance evaluation. These KPIs include various aspects such as semantic computation consumption, quality of semantic information extraction, and semantic capacity. Current research mainly employs two methodologies to derive KPIs in semantic communication. The first approach relies on simulation, where semantic-related metrics, such as semantic rate, are obtained utilizing functions derived from simulation results \cite{9763856,9953095,hu2023multiuser,9398576}. The second approach involves analysis, where expressions related to semantic communication, such as semantic computation consumption, are derived through theoretical analysis \cite{10333452,yang2023secure,zhao2023joint,yang2023energy}. In simulation-based studies, Yan et al. achieved maximum spectral efficiency by optimizing channel assignment and the number of semantic symbols \cite{9763856,10001594}. Addressing energy efficiency, the authors in \cite{10012845} conducted optimization for total energy consumption under latency constraints. Cang et al. integrated semantic communication with mobile edge computing (MEC), minimizing energy consumption by optimizing semantic-aware division factors and managing communication and computation resources \cite{cang2023resource}. In analysis-based studies, the authors in \cite{10333452} optimized the total energy of the entire system through strategic semantic level selections.

In addition to characterizing the KPIs of semantic communication, the representation of semantic information is also a challenging aspect of semantic communication \cite{10183794}. Although many approaches use auto-encoders for semantic compression \cite{9953076,9953316,9450827}, resulting in data of small size that is considered to be semantic information, this output often lacks interpretability and cannot be directly validated by interaction with human understanding. To address this limitation, some works \cite{10061867,9685056} proposed the use of knowledge graphs as a representation method aligned with human logic. A knowledge graph generally consists of a set of nodes connected by edges \cite{9357868}. Each node represents an entity, which can be a real-world object, a concept, a temporal reference, etc. The edges represent the semantic relationship between these entities. An illustrative example of a knowledge graph is shown in Fig.~\ref{fg:kg}. Notably, knowledge graphs efficiently encapsulate substantial information within a compact data size, making them an ideal candidate for semantic information representation.

\begin{figure}[t]
\centering
\includegraphics[width=0.9\linewidth]{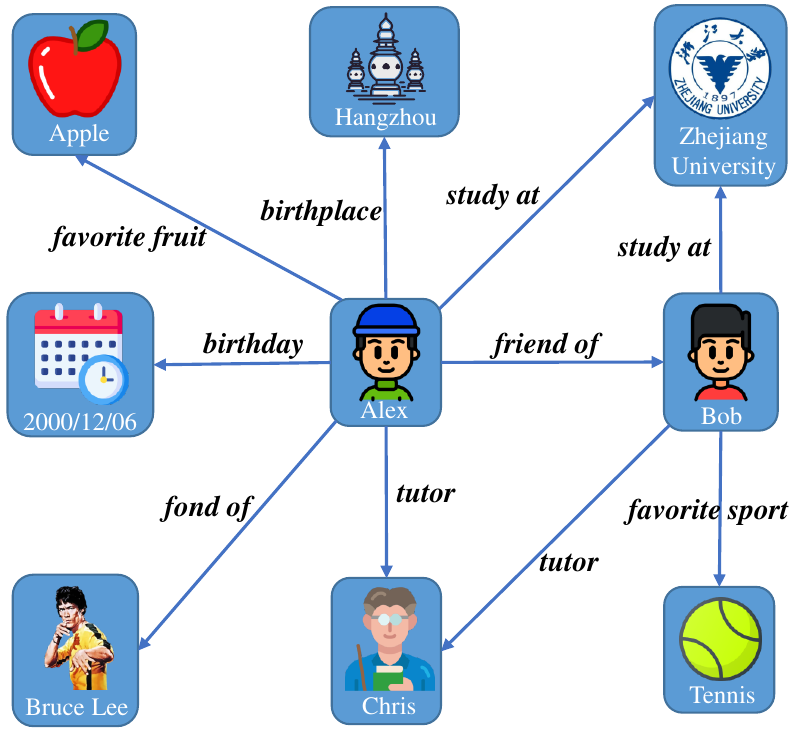}
\caption{Illustration of a knowledge graph.}
\label{fg:kg}
\end{figure}

Recently, there has been significant research investigating semantic communication over wireless networks. The authors in \cite{8461983} introduced deep learning techniques to joint source-channel coding of text, which laid the foundation of a semantic communication system for text transmission. This research offered novel perspectives and methods for effectively encoding and transmitting textual information. Building upon this, Yao et al. further explored the design of text transmission by proposing an iterative semantic coding approach \cite{9834044}. The objective of this approach was to accurately capture and transmit the semantic content of text, thereby enhancing the efficiency and accuracy of transmission. Further, semantic triples and knowledge graphs have been employed to enable semantic communication. Liu et al. investigated a task-oriented semantic communication approach based on semantic triples \cite{10118916}. This approach focused on effectively encoding and transmitting key semantic information based on specific task requirements. Additionally, the work in \cite{9838470} proposed a cognitive semantic communication framework with knowledge graphs. This work presented a simple, general, and interpretable solution for detecting semantic information by utilizing triples as semantic symbols. Considering the unique property of semantic communication, resource allocation and performance optimization are crucial factors to consider in the development of semantic communication systems. Wang et al. employed deep reinforcement learning to address the resource allocation problem in semantic communication \cite{9832831}. This study introduced new strategies to effectively allocate communication resources to ensure efficient transmission of semantic information. However, the aforementioned works \cite{8461983,9834044,10118916,9838470,9832831} did not take into account the computational power requirements of semantic communication systems, which is important for energy-constrained wireless networks \cite{6861946}.

In this paper, we develop a multi-user probabilistic semantic communication (PSC) framework that jointly considers transmission and computation consumption. The key contributions of this work are summarized as follows:
\begin{itemize}
\item We consider a PSC network in which multiple users employ semantic information extraction techniques to compress their original large-sized data and transmit the extracted information to a multi-antenna base station (BS). In our model, users' large-sized data is represented by extensive knowledge graphs and is compressed based on the shared probability graph between the users and the BS.
\item We formulate an optimization problem that aims to maximize the sum equivalent rate of all users while considering total power and semantic resource limit constraints. This joint optimization problem takes into account the trade-off between transmission efficiency and computation complexity.
\item To solve this non-convex non-smooth problem, a low-complexity three-stage algorithm is proposed. In stage 1, the receive beamforming matrix is optimized using the minimum mean square error (MMSE) strategy. In stage 2, we substitute the transmit power with the semantic compression ratio and develop an alternating optimization (AO) method to perform a rough search for the semantic compression ratio. In stage 3, gradient ascent is used to refine the semantic compression ratio. Numerical results show the effectiveness of the proposed algorithm.
\end{itemize}

The remainder of this paper is organized as follows. The system model and problem formulation are described in Section \uppercase\expandafter{\romannumeral2}. The algorithm design is presented in Section \uppercase\expandafter{\romannumeral3}. Simulation results are analyzed in Section \uppercase\expandafter{\romannumeral4}. Conclusions are drawn in Section \uppercase\expandafter{\romannumeral5}.

\section{System Model and Problem Formulation}

\begin{figure}[t]
\centering
\includegraphics[width=\linewidth]{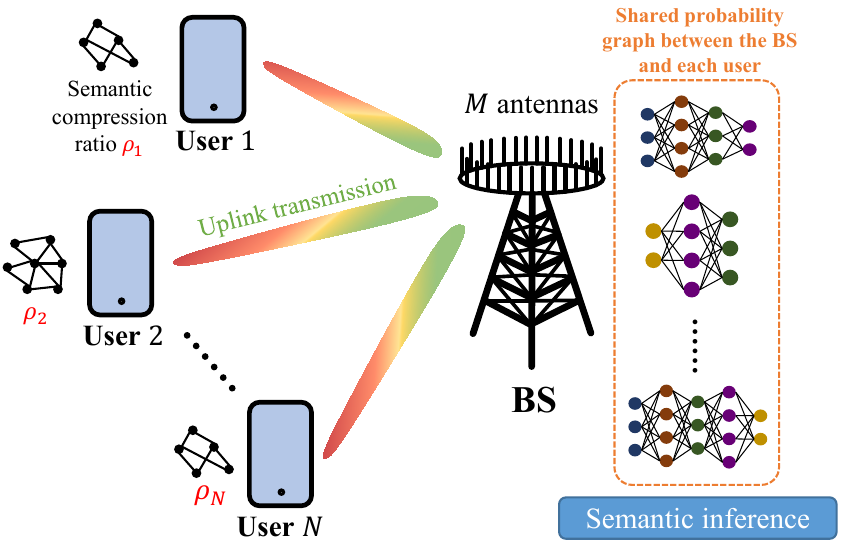}
\caption{An illustration of the considered PSC network.}
\label{fg:sm}
\end{figure}

Consider an uplink wireless PSC network with one multi-antenna BS and $N$ single-antenna users, as shown in Fig.~\ref{fg:sm}. The BS is equipped with $M$ antennas, and the set of users is represented by $\mathcal{N}$. Each user, denoted by $n$, has a large-sized data $\mathcal{D}_n$ to be transmitted. Due to limited wireless resource, the users need to extract the small-sized semantic information $\mathcal{C}_n$ from the original data $\mathcal{D}_n$. In the considered model, users first extract the semantic information based on their individual local probability graphs and then transmit the semantic data to the BS.

\subsection{Semantic Communication Model}

\begin{figure}[t]
\centering
\includegraphics[width=\linewidth]{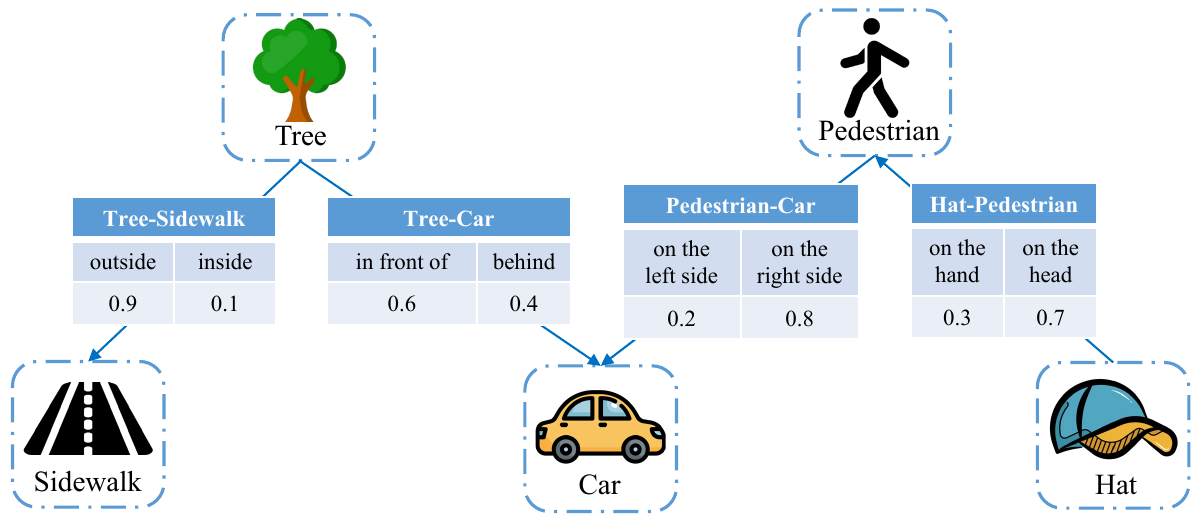}
\caption{Illustration of the probability graph considered in the PSC system.}
\label{fg:pg}
\end{figure}

 We employ probability graphs as the knowledge base between the semantic transmitter (each user) and the semantic receiver (BS). A probability graph integrates information from multiple knowledge graphs, extending the conventional knowledge graph by introducing the dimension of relational probability. An illustrative example of a probability graph is depicted in Fig.~\ref{fg:pg}. A traditional knowledge graph comprises numerous triples, and each triple can be represented by
\begin{equation}\label{triple}
    \varepsilon = (h, r, t),
\end{equation}
where $h$ is the head entity, $t$ denotes the tail entity, and $r$ represents the relation between $h$ and $t$. In a traditional knowledge graph, the relations are typically fixed. In contrast, in a probability graph, each relation is associated with a specific probability, representing the likelihood of that particular relation occurring under the given conditions of fixed head entity and tail entity.

We assume that each user needs to transmit several knowledge graphs. These knowledge graphs are generated from extensive textual data (picture/audio/video data can also be applied) after undergoing named entity recognition (NER) \cite{9039685} and relation extraction (RE) \cite{9446853}, resulting in abstracted information. Using the shared probability graph between a user and the BS, one can further compress the transmitted knowledge graphs.

The probability graph extends the dimensionality of relations by statistically enumerating the occurrences of various relations associated with the same head and tail entities across diverse knowledge graph samples. Leveraging the statistical information from the probability graph, a multidimensional conditional probability matrix can be constructed. This matrix reflects the likelihood of a specific triple being valid under the condition that certain other triples are valid. This enables the omission of relations in the knowledge graph before transmission, resulting in data compression. However, it is crucial to note that achieving a smaller data size necessitates a lower semantic compression ratio, which demands higher-dimensional conditional probabilities. This decrease in semantic compression ratio comes at the cost of increased computational load, thus presenting a trade-off between communication and computation for the considered PSC network. The specific implementation details of the probability graph can be found in \cite{10333452}.

\begin{figure*}[t]
\centering
\includegraphics[width=\linewidth]{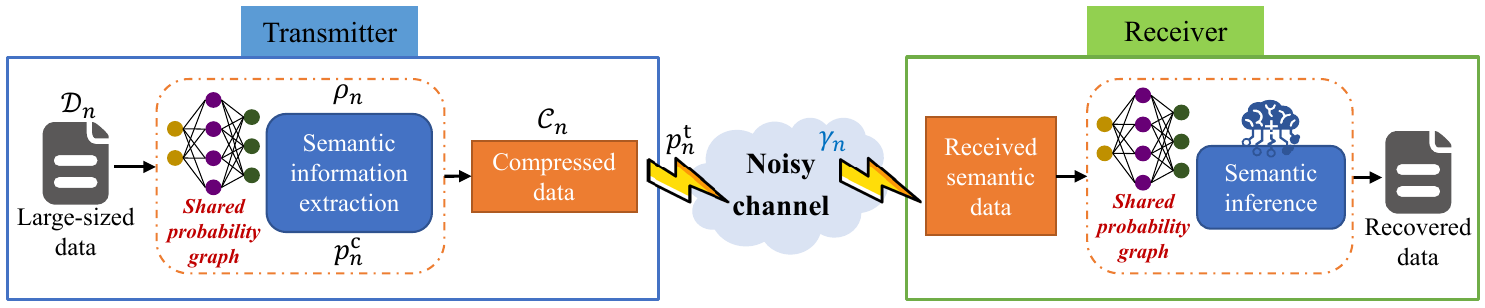}
\caption{The framework of considered PSC network.}
\label{fg:fw}
\end{figure*}

Within the framework of the considered PSC network, each user possesses a personalized local probability graph that stores statistical information about their historical data. Each user $n$ individually performs semantic information extraction, compressing original large-sized data $\mathcal{D}_n$ based on its stored probability graph with the semantic compression ratio denoted by $\rho_n$. Subsequently, the obtained compressed data, $\mathcal{C}_n$, is transmitted to the BS with transmit power $p_n^\mathrm{t}$. Meanwhile, the BS maintains identical probability graphs corresponding to all $N$ users. Once the BS receives the semantic data from user $n$, it conducts semantic inference to recover the compressed semantic information using the shared probability graph of user $n$. The overall framework of the considered PSC network is depicted in Fig.~\ref{fg:fw}.

\subsection{Transmission Model}
As mentioned above, the BS is equipped with $M$ antennas to serve $N$ single-antenna users. We assume that the number of users is not greater than the number of antennas in the BS, that is, $N\leq M$. Therefore, space-division multiple access (SDMA) can be employed.

We consider the uplink transmission from all users to the BS, and the received signal at the BS can be mathematically represented by
\begin{equation}\label{rs}
    \mathbf{y}=\mathbf{W}^\mathrm{H}\mathbf{H}\mathbf{x}+\mathbf{W}^\mathrm{H}\mathbf{n},
\end{equation}
where $\mathbf{W}=[\mathbf{w}_1,\mathbf{w}_2,\cdots,\mathbf{w}_N]\in \mathbb{C}^{M\times N}$ represents the receive beamforming matrix at the BS, with $\mathbf{w}_n\in \mathbb{C}^{M\times 1}$ being the receive beamforming vector for user $n$. The matrix $\mathbf{H}=[\mathbf{h}_1,\mathbf{h}_2,\cdots,\mathbf{h}_N]\in \mathbb{C}^{M\times N}$ denotes the multiple access channel matrix from all $N$ users to the antenna array of the BS. Each vector $\mathbf{h}_n\in \mathbb{C}^{M\times 1}$ represents the channel vector between the BS and user $n$, and is determined by the specific propagation environment. Here, we assume $[\mathbf{H}]_{i,j}\sim\mathcal{CN}(0,\beta)$ where $[\mathbf{\cdot}]_{i,j}$ denotes an element in a matrix and $\beta$ signifies the long-term channel power gain. The vector $\mathbf{x}=[x_1,x_2,\cdots,x_N]^\mathrm{T}\in \mathbb{C}^{N\times 1}$ denotes the transmitted signals of the users with transmit power $\mathbf{p}=[p_1^\mathrm{t},p_2^\mathrm{t},\cdots,p_N^\mathrm{t}]^\mathrm{T}$, where the transmit power of user $n$ is denoted by $p_n^\mathrm{t}$. The vector $\mathbf{n}=[n_1,n_2,\cdots,n_M]^\mathrm{T}$ represents additive white Gaussian noise (AWGN) at the BS. We assume that $[\mathbf{n}]_{i}\sim\mathcal{CN}(0,\sigma^2)$, where $[\mathbf{\cdot}]_{i}$ denotes an element in a vector, and $\sigma^2$ denotes the average noise power.

For the uplink transmission that utilizes linear combining at the BS, the received signal-to-interference-plus-noise ratio (SINR) for the signal from user $n$ can be given by
\begin{equation}\label{sinr}
    \gamma_n = \frac{\left\vert \mathbf{w}^\mathrm{H}_n \mathbf{h}_n\right\vert^2 p_n^\mathrm{t}}{\sum\limits^N_{k=1,k\neq n}\left\vert \mathbf{w}^\mathrm{H}_n \mathbf{h}_k\right\vert^2 p_k^\mathrm{t} + \left\Vert \mathbf{w}_n\right\Vert^2_2\sigma^2},
\end{equation}
and the achievable rate of user $n$ can be expressed as
\begin{equation}\label{ar}
    C_n=\log_2(1+\gamma_n).
\end{equation}

In the considered PSC network, the original large-sized data $\mathcal D_n$ is compressed into a small-sized data $\mathcal C_n$ with a semantic compression ratio prior to transmission. The semantic compression ratio for user $n$ is defined as
\begin{equation}\label{cr}
    \rho_n=\frac{\mathrm{size}(\mathcal{C}_n)}{\mathrm{size}(\mathcal{D}_n)},
\end{equation}
where the function $\mathrm{size}(\cdot)$ quantifies the data size in terms of bits.

Hence, we can calculate an equivalent rate for user $n$, denoted by
\begin{equation}\label{er}
    R_n = \frac{1}{\rho_n}C_n,
\end{equation}
which represents the transmission rate perceived by the receiver following the process of decoding. Due to the fact that one bit in the compressed data $\mathcal C_n$ can represent $1/\rho_n$ bits in the original data $\mathcal D_n$, we multiply the factor $1/\rho_n$ in equivalent expression \eqref{er}.

\subsection{Computation Model}
Each user $n$ needs to perform semantic information extraction based on their local probability graph to compress the original data $\mathcal{D}_n$ into a smaller-sized data $\mathcal{C}_n$. This operation relies on computational resources, and it is important to note that the lower the semantic compression ratio $\rho_n$, the higher the computation load becomes.

According to equation (19) in \cite{10333452}, the computation load for the considered probability graph-based PSC network can be expressed as
\begin{equation}\label{cl}
    g\left(\rho\right)=\left\{\begin{array}{l}
        A_1\rho +B_1, L_1< \rho \leq 1, \\
        A_2\rho +B_2, L_2< \rho \leq L_1, \\
        \vdots \\
        A_S\rho +B_S, L_S\leq \rho \leq L_{S-1}.
    \end{array}\right.,
\end{equation}
where $A_s<0$ represents the slope, $B_s>0$ stands for the constant term, and $L_s$ is the boundary for each segment $s=1, 2, \cdots, S$. These parameters are system-specific and are determined by the characteristics of the probability graphs. From \eqref{cl}, the computation load expression is a piecewise function, which is due to the fact that the semantic inference involves multiple levels of conditional probability functions and each level of conditional probability function results in one linear computation load expression.

\begin{figure}[t]
\centering
\includegraphics[width=\linewidth]{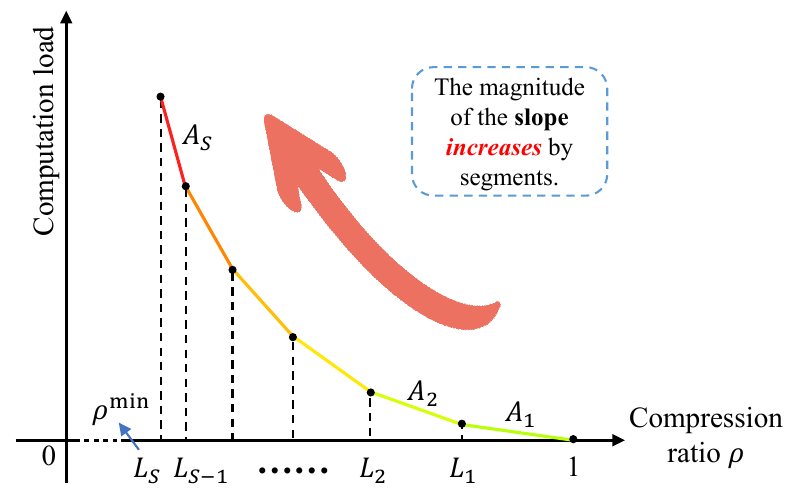}
\caption{Illustration of computation load versus semantic compression ratio $\rho$.}
\label{fg:cl}
\end{figure}

Based on \eqref{cl}, the computation load, denoted by $g(\rho)$, exhibits a segmented structure with $S$ levels, and the slope magnitude decreases in discrete segments, as depicted in Fig.~\ref{fg:cl}. This is because when the compression ratio is high, only low-dimensional conditional probabilities are employed, resulting in lower computational demands. However, as the compression ratio decreases, the need for higher-dimensional information arises. With higher information dimensions, the computation load becomes more intensive. Each transition in the segmented function $g(\rho)$ represents the utilization of probabilistic information with more information for semantic information extraction.

Given the piecewise property of the computation load function, the computation power of user $n$ can be written as
\begin{equation}\label{cp}
    p_n^\mathrm{c} = g_n(\rho_n)p_0,
\end{equation}
where $p_0$ represents a positive constant denoting the computation power coefficient, 
$g_n(\rho_n)=A_{ns} \rho_{n}+B_{ns}$, if $L_{ns}\leq \rho_n \leq L_{n(s-1)}$, $\forall s=1, 2, \cdots, S$, and $L_{ns}<L_{n(s-1)}<\cdots<L_{n1}<L_{n0}=1$.
%%%%%

In this paper, our primary focus is on the computation load at the user side, as we are specifically addressing the uplink transmission scenario. In this context, each user needs to perform an information transmission task, and as such, the computational overhead associated with semantic decoding at the BS is ignored since the BS always has high power budget.

\subsection{Problem Formulation}
Given the considered system model, our objective is to maximize the sum of equivalent rate for all users through jointly optimizing semantic compression ratio of each user, transmit power of each user, and receive beamforming matrix of the BS while considering the maximum total power of each user. The sum rate maximization problem can be formulated as
\begin{subequations}\label{pf}
    \begin{align}
        \max_{\bm{\rho},\mathbf{p},\mathbf{W}} \quad & \sum_{n=1}^N R_n, \tag{\ref{pf}}\\
        \textrm{s.t.} \quad & p_n^\mathrm{t}+p_n^\mathrm{c}\leq p_n^\mathrm{max},\forall n\in\mathcal{N}, \label{c1}\\
        & p_n^\mathrm{t}\geq 0,\forall n\in\mathcal{N}, \label{c2}\\
        & \rho_{n}^\mathrm{min}\leq\rho_n\leq 1,\forall n\in\mathcal{N}, \label{c3}
    \end{align}
\end{subequations}
where $\bm{\rho}=[\rho_1,\rho_2,\cdots,\rho_N]^\mathrm{T}$, $\mathcal N=\{1, 2, \cdots, N\}$, and $\rho_{n}^\mathrm{min}$ is the semantic compression limit for user $n$. Constraint \eqref{c1} reflects a limit on the sum of transmit power and computation power for user $n$, ensuring it remains within the overall power limit $p_n^\mathrm{max}$. Constraint \eqref{c2} enforces the non-negativity of user's transmit power. Lastly, constraint \eqref{c3} bounds the semantic compression ratio for each user.

It is essential to recognize that semantic compression ratio and transmit power are tightly coupled in problem \eqref{pf}. Smaller compression ratios lead to larger values of the objective function, but the presence of constraint \eqref{c1} limits the transmit power, consequently reducing the objective function.
% Conversely, higher transmit powers result in larger objective function values, but the presence of constraint \eqref{c1} elevates the compression ratio, leading to a decrease in the objective function.
Therefore, achieving the right balance between the effects of semantic compression ratio and transmit power is the key to the solution of problem \eqref{pf}. Another important aspect of problem \eqref{pf} is the inclusion of the segmented function $g_n(\rho_n)$ in constraint \eqref{c1}, which introduces a distinct challenge to the optimization process. Since the objective function is highly non-convex and constraint \eqref{c1} is non-smooth, it is generally hard to obtain the optimal solution of problem \eqref{pf} with existing optimization tools in polynomial time. Thus, we develop a suboptimal solution in the next section.

\section{Algorithm Design}
In this section, a three-step algorithm is proposed to solve problem \eqref{pf}, i.e., MMSE for receive beamforming matrix, rough search for semantic compression ratio, and refined search for semantic compression ratio. These three stages will be explained in detail below.

\subsection{Stage 1: MMSE for Receive Beamforming Matrix}
With the advancement of multiple-input multiple-output (MIMO) technology, various beamforming methods, including maximum ratio combining (MRC), zero-forcing (ZF), and MMSE, have been developed to deal with multi-user interference. In this section, we employ MMSE strategy to identify the receive beamforming matrix $\mathbf{W}$, which is effective in dealing with the high noise power situations. Based on the MMSE technique, the closed-form solution of receive beamforming matrix $\mathbf{W}$ is given in the following lemma.
\begin{lemma}\label{lemma1}
For any given transmit power of each user, i.e., $\mathbf{p}$, the optimal linear receive beamforming matrix $\mathbf{W}$ of the BS under MMSE strategy can be written as
\begin{equation}\label{mmse}
    \mathbf{W}(\mathbf{P}) = \left(\mathbf{H}\mathbf{P}\mathbf{H}^\mathrm{H}+\sigma^2\mathbf{I}_M\right)^{-1}\mathbf{H}\mathbf{P},
\end{equation}
where $\mathbf{P}=\mathrm{diag}\{\mathbf{p}\}$ represents a diagonal matrix with $[\mathbf{P}]_{i,i}=[\mathbf{p}]_{i}$, and $\mathbf{I}_M$ is an identical matrix of size $M\times M$.
\end{lemma}

\begin{proof}
See Appendix A.
$\hfill\square$
\end{proof}

According to Lemma \ref{lemma1}, the optimal MMSE receive beamforming is obtain as a closed-form solution, which is a function of the transmit power of all users. Based on the obtained $\mathbf{W}(\mathbf{P})$, we have
\begin{equation}\label{w}
    \mathbf{w}_n = p_n^\mathrm{t}\left(\mathbf{H}\mathbf{P}\mathbf{H}^\mathrm{H}+\sigma^2\mathbf{I}_M\right)^{-1}\mathbf{h}_n.
\end{equation}

For notation convenience, we define
\begin{equation}\label{u}
    U_{nk}\triangleq \left\vert \mathbf{w}^\mathrm{H}_n \mathbf{h}_k\right\vert^2= \left(p_n^\mathrm{t}\right)^2\left\vert \mathbf{h}_n^\mathrm{H}\left(\mathbf{H}\mathbf{P}\mathbf{H}^\mathrm{H}+\sigma^2\mathbf{I}_M\right)^{-1}\mathbf{h}_k\right\vert^2,
\end{equation}
and
\begin{equation}\label{v}
    v_n\triangleq \left\Vert \mathbf{w}_n\right\Vert^2_2\sigma^2 = \left(p_n^\mathrm{t}\sigma\right)^2\left\Vert \left(\mathbf{H}\mathbf{P}\mathbf{H}^\mathrm{H}+\sigma^2\mathbf{I}_M\right)^{-1}\mathbf{h}_n\right\Vert^2_2.
\end{equation}
Thus, by substituting \eqref{w} into \eqref{sinr}, the received SINR for the signal from user $n$ can be rewritten as
\begin{equation}\label{newsinr}
    \gamma_n = \frac{U_{nn} p_n^\mathrm{t}}{\sum\limits^N_{k=1,k\neq n}U_{nk} p_k^\mathrm{t} + v_n}.
\end{equation}

With the above variable substitution, problem \eqref{pf} can be reformulated as
\begin{subequations}\label{s1}
    \begin{align}
        \max_{\bm{\rho},\mathbf{p}} \quad & \sum_{n=1}^N \frac{1}{\rho_n}\log_2\left(1+\frac{U_{nn} p_n^\mathrm{t}}{\sum\limits^N_{k=1,k\neq n}U_{nk} p_k^\mathrm{t} + v_n}\right), \tag{\ref{s1}}\\
        \textrm{s.t.} \quad & p_n^\mathrm{t}+p_n^\mathrm{c}\leq p_n^\mathrm{max},\forall n\in\mathcal{N}, \label{c11}\\
        & p_n^\mathrm{t}\geq 0,\forall n\in\mathcal{N}, \label{c12}\\
        & \rho_{n}^\mathrm{min}\leq\rho_n\leq 1,\forall n\in\mathcal{N}. \label{c13}
    \end{align}
\end{subequations}

In this stage, the receive beamforming matrix $\mathbf{W}$ is optimized using MMSE strategy with a closed-form solution. Hence, the variables that require optimization in problem \eqref{pf} are reduced, and the problem we need to solve becomes problem \eqref{s1}.

\subsection{Stage 2: Rough Search for Semantic Compression Ratio}
In stage 2, we will roughly determine the semantic compression ratio $\rho_n$ for each user by identifying the segment in the piecewise function $g_n(\rho_n)$ where $\rho_n$ falls.

Without loss of generality, it is assumed that when the semantic compression ratio is equal to $\rho_{n}^\mathrm{min}$, the computation power $p_n^\mathrm{c}$ exceeds the total power limit $p_n^\mathrm{max}$, i.e.,
\begin{equation}\label{assumption}
    g_n(\rho_{n}^\mathrm{min})p_0\geq p_n^\mathrm{max},\forall n\in\mathcal{N}.
\end{equation}
This is because as the semantic compression ratio tends to $\rho_{n}^\mathrm{min}$, the computation load rises dramatically as the probability dimension of the computation becomes very high.

With the above assumption, the following theorem can be derived.
\begin{theorem}\label{theorem1}
The optimal semantic compression ratio $\rho_n^*$ and transmit power $\left(p_n^\mathrm{t}\right)^*$ of problem \eqref{s1} must satisfy
\begin{equation}\label{theorem}
    \left(p_n^\mathrm{t}\right)^*+g_n(\rho_n^*)p_0=p_n^\mathrm{max},\forall n\in\mathcal{N}.
\end{equation}
\end{theorem}

\begin{proof}
See Appendix B.
$\hfill\square$
\end{proof}

Theorem \ref{theorem1} implies that constraint \eqref{c11} will always hold with equality for optimality of problem \eqref{s1}. Based on Theorem \ref{theorem1}, we can substitute $p_n^\mathrm{t}=p_n^\mathrm{max}-g_n(\rho_n)p_0$ into problem \eqref{s1}. Thus, problem \eqref{s1} can be rewritten as
\begin{subequations}\label{s2}
    \begin{align}
        \max_{\bm{\rho}} \quad & \sum_{n=1}^N \frac{1}{\rho_n}\log_2\Bigg(1+\notag\\
        &\qquad\qquad \frac{U_{nn} \left[p_n^\mathrm{max}-g_n(\rho_n)p_0\right]}{\sum\limits^{N}_{k=1,k\neq n}U_{nk} \left[p_k^\mathrm{max}-g_n(\rho_k)p_0\right] + v_n}\Bigg), \tag{\ref{s2}}\\
        \textrm{s.t.} \quad & p_n^\mathrm{max}-g_n(\rho_n)p_0\geq 0,\forall n\in\mathcal{N}, \label{c21}\\
        & \rho_{n}^\mathrm{min}\leq\rho_n\leq 1,\forall n\in\mathcal{N}. \label{c22}
    \end{align}
\end{subequations}
Note that $U_{nk}$ and $v_n$ are variables associated with the transmit power $\mathbf{p}$ according to equations \eqref{u} and \eqref{v}. Since transmit power $p_n^\mathrm{t}$ is also a function of the semantic compression ratio $\rho_n$, $U_{nk}$ and $v_n$ become variables only associated with the semantic compression ratio $\bm{\rho}$. Therefore, problem \eqref{s2} is related solely to the semantic compression ratio.

However, the difficulty in solving problem \eqref{s2} still exists due to the non-convexity of the objective function and the non-smoothness of the computation load function, $g_n(\rho_n)$. To handle the non-smoothness of $g_n(\rho_n)$, it can be reformulated as
\begin{equation}\label{grho}
    g_n(\rho_n)=\sum_{s=1}^{S} \theta_{ns}(A_{ns}\rho_n+B_{ns}),\theta_{ns}\in \{0,1\},\sum_{s=1}^{S}\theta_{ns}=1,
\end{equation}
where $S$ is the number of segments of the piecewise function $g_n(\rho_n)$, and $\theta_{ns}$ identifies the specific segment within which $\rho_n$ falls.

Therefore, problem \eqref{s2} can be rewritten as
\begin{subequations}\label{s22}
    \begin{align}
        \max_{\bm{\Theta},\bm{\rho}} \quad & \sum_{n=1}^N \frac{1}{\rho_n}\log_2\Bigg(1+\notag\\
        & \frac{U_{nn} \left[p_n^\mathrm{max}-p_0\sum\limits_{s=1}^{S} \theta_{ns}(A_{ns}\rho_n+B_{ns})\right]}{\sum\limits^{N}_{k=1,k\neq n}U_{nk} \left[p_k^\mathrm{max}-p_0\sum\limits_{s=1}^{S} \theta_{ks}(A_{ks}\rho_k+B_{ks})\right] + v_n}\Bigg), \tag{\ref{s22}}\\
        \textrm{s.t.} \quad & \sum_{s=1}^{S} \theta_{ns}(A_{ns}\rho_n+B_{ns})\leq \frac{p_n^\mathrm{max}}{p_0},\forall n\in\mathcal{N}, \label{c221}\\
        & \rho_{n}^\mathrm{min}\leq\rho_n\leq 1,\forall n\in\mathcal{N}, \label{c222}\\
        & \sum_{s=1}^{S}\theta_{ns}=1,\forall n\in\mathcal{N}, \label{c223}\\
        & \theta_{ns}\in \{0,1\},\forall n\in\mathcal{N}, \label{c224}
    \end{align}
\end{subequations}
where $\bm{\Theta}=[\bm{\theta}_1,\bm{\theta}_2,\cdots,\bm{\theta}_N]$, and $\bm{\theta}_n=[\theta_{n1},\theta_{n2},\cdots,\theta_{nS}]^\mathrm{T}$.

In problem \eqref{s22}, both binary integer matrix $\boldsymbol \Theta$ and continuous variable $\boldsymbol \rho$ are involved. Thus, problem \eqref{s22} becomes a challenging mixed-integer programming problem.

It is important to note that $\bm{\Theta}$ and $\bm{\rho}$ are highly coupled in objective function \eqref{s22} and constraint \eqref{c221}. If $\bm{\rho}$ is determined, then so is $\bm{\Theta}$. However, a determined $\bm{\Theta}$ cannot result in a determined $\bm{\rho}$, but it can narrow down the possible range of $\bm{\rho}$ by specifying the particular segment in which $\bm{\rho}$ exists.

Therefore, we obtain an approximate estimation of the semantic compression ratio $\bm{\rho}$ by determining $\bm{\Theta}$ as follows.

For convenience, we define
\begin{equation}\label{midrho}
    \rho_{ns}=\frac{L_{n(s-1)}+L_{ns}}{2},1\leq s\leq S,
\end{equation}
which represents the middle value of the semantic compression ratio in segment $s$ for user $n$. 

We can see that $\rho_{ns}$ is a fixed value denoting the midpoint of segment $s$ in $g_n(\rho_n)$. Therefore, we use $\rho_{ns}$ for approximating the value of $\rho_n$ in every segment $s$. By making this approximation, problem \eqref{s22} can be simplified as
\begin{subequations}\label{s23}
    \begin{align}
        \max_{\bm{\Theta}} \quad & \sum_{n=1}^N \frac{1}{\sum_{s=1}^{S}\theta_{ns}\rho_{ns}}\log_2\Bigg(1+\notag\\
        & \frac{U_{nn} \left[p_n^\mathrm{max}-p_0\sum\limits_{s=1}^{S} \theta_{ns}(A_{ns}\rho_{ns}+B_{ns})\right]}{\sum\limits^{N}_{k=1,k\neq n}U_{nk} \left[p_k^\mathrm{max}-p_0\sum\limits_{s=1}^{S} \theta_{ks}(A_{ks}\rho_{ks}+B_{ks})\right] + v_n}\Bigg), \tag{\ref{s23}}\\
        \textrm{s.t.} \quad & \sum_{s=1}^{S} \theta_{ns}(A_{ns}\rho_{ns}+B_{ns})\leq \frac{p_n^\mathrm{max}}{p_0},\forall n\in\mathcal{N}, \label{c231}\\
        & \sum_{s=1}^{S}\theta_{ns}=1,\forall n\in\mathcal{N}, \label{c232}\\
        & \theta_{ns}\in \{0,1\},\forall n\in\mathcal{N}. \label{c233}
    \end{align}
\end{subequations}
Problem \eqref{s23} is an integer programming problem with respect to the Boolean matrix $\bm{\Theta}$.

Since the objective function of problem \eqref{s23} remains intractable and challenging to convert into a convex function, we present an AO method to iteratively determine the integer matrix $\bm{\Theta}$.

With given semantic compression ratio level indicating vectors of other $N-1$ users, we need to determine the optimal $\bm{\theta}_n$ for the current user $n$. Then, we can have the following problem
\begin{subequations}\label{s24}
    \begin{align}
        \max_{\bm{\theta}_n} \quad & \sum_{n=1}^N \frac{1}{\sum_{s=1}^{S}\theta_{ns}\rho_{ns}}\log_2\Bigg(1+\notag\\
        & \frac{U_{nn} \left[p_n^\mathrm{max}-p_0\sum\limits_{s=1}^{S} \theta_{ns}(A_{ns}\rho_{ns}+B_{ns})\right]}{\sum\limits^{N}_{k=1,k\neq n}U_{nk} \left[p_k^\mathrm{max}-p_0\sum\limits_{s=1}^{S} \theta_{ks}(A_{ks}\rho_{ks}+B_{ks})\right] + v_n}\Bigg), \tag{\ref{s24}}\\
        \textrm{s.t.} \quad & \sum_{s=1}^{S} \theta_{ns}(A_{ns}\rho_{ns}+B_{ns})\leq \frac{p_n^\mathrm{max}}{p_0},\forall n\in\mathcal{N}, \label{c241}\\
        & \sum_{s=1}^{S}\theta_{ns}=1,\forall n\in\mathcal{N}, \label{c242}\\
        & \theta_{ns}\in \{0,1\},\forall n\in\mathcal{N}. \label{c243}
    \end{align}
\end{subequations}

Since $\bm{\theta}_n$ is a one-hot vector of size $S\times 1$, we can simply iterate through all the possible locations where `1' could occur, which has $S$ possibilities. The $\bm{\theta}_n$ corresponding to the maximum objective function value is saved for subsequent iterations.

The iteration terminates when the objective function value of problem \eqref{s24} converges or the iteration count reaches the maximum limit of $I^\mathrm{max}$.
% The initial solution is set to $\bm{\theta}=[1,0,0,\cdots,0]$ for all users. 
Algorithm \ref{algo1} summarizes the AO method for solving the integer programming problem \eqref{s23}.

\begin{algorithm}[ht]
\caption{Alternating Optimization for Determining Integer Matrix $\bm{\Theta}$}\label{algo1}
\begin{algorithmic}[1]
    \STATE Initialize $\bm{\Theta}^{(0)}$. Set iteration index $i=0$.
    \REPEAT
        \FOR{$n=1$ to $N$}
            \FOR{$s=1$ to $S$}
            \IF{Constraint \eqref{c241} is satisfied}
            \STATE Calculate the objective value for $\theta_{ns}=1$, $\theta_{nt}=0$, $\forall t\neq s$.
            \ELSE
            \STATE Set the objective value as zero.
            \ENDIF
            \ENDFOR
        \STATE Update $\bm{\theta}_n$ which corresponds to the maximum objective value.
        \ENDFOR
        \STATE Obtain $\bm{\Theta}^{(i+1)}$.
        \STATE Set $i=i+1$.
    \UNTIL{the objective value of problem \eqref{pf} converges or $i>I^\mathrm{max}$.}
    \STATE \textbf{Output}: The optimized Boolean matrix $\bm{\Theta}$.
\end{algorithmic}
\end{algorithm}

In this stage, the transmit power $\mathbf{p}$ is substituted with the semantic compression ratio $\bm{\rho}$ according to Theorem \ref{theorem1}. Furthermore, the matrix $\bm{\Theta}$, which determines the range of $\rho_n$ for each user, is optimized employing the AO method. Next, we need to perform a refined search for the semantic compression ratio $\bm{\rho}$.

\subsection{Stage 3: Refined Search for Semantic Compression Ratio}
To achieve an accurate value for the semantic compression ratio, a refined search is required in stage 3. This is because the result obtained in stage 2 is only an approximate estimate of the semantic compression ratio.

Based on the Boolean matrix $\bm{\Theta}$ obtained in stage 2, we can determine the segment in which $\bm{\rho}$ falls. Denote the selected segment for user $n$ by $S_n$, which means
\begin{equation}\label{segment}
    g_n(\rho_n)=A_{n(S_n)}\rho_n+B_{n(S_n)},L_{n(S_n)}\leq\rho_n\leq L_{n(S_n-1)}.
\end{equation}
Once the segment of $\rho_n$ is determined, the computation load function $g_n(\rho_n)$ becomes a linear function instead of a non-smooth piecewise function.

Therefore, the problem need to solve in stage 3 can be reformulated as
\begin{subequations}\label{s3}
    \begin{align}
        \max_{\bm{\rho}} \quad & \sum_{n=1}^N \frac{1}{\rho_n}\log_2\Bigg(1+\notag\\
        & \frac{U_{nn} \left[p_n^\mathrm{max}-p_0\left(A_{n(S_n)}\rho_n+B_{n(S_n)}\right)\right]}{\sum\limits^{N}_{k=1,k\neq n}U_{nk} \left[p_k^\mathrm{max}-p_0\left(A_{k(S_k)}\rho_k+B_{k(S_k)}\right)\right] + v_n}\Bigg), \tag{\ref{s3}}\\
        \textrm{s.t.} \quad & A_{n(S_n)}\rho_n+B_{n(S_n)}\leq \frac{p_n^\mathrm{max}}{p_0},\forall n\in\mathcal{N}, \label{c31}\\
        & L_{(S_n)}\leq\rho_n\leq L_{n(S_n-1)},\forall n\in\mathcal{N}. \label{c32}
    \end{align}
\end{subequations}
Problem \eqref{s3} is no longer non-smooth as the piecewise function $g_n(\rho_n)$ has been degraded to a linear function. However, problem \eqref{s3} remains non-convex as the objective function is highly non-convex with respect to $\bm{\rho}$. Thus, it is generally hard to obtain the globally optimal solution for problem \eqref{s3}. Next, we employ the gradient ascent method to obtain a suboptimal solution.

For convenience, we define
\begin{subequations}\label{obj}
    \begin{multline}
        f\left(\bm{\rho}\right)=\sum_{n=1}^N \frac{1}{\rho_n}\log_2\Bigg(1+\notag\\
            \frac{U_{nn} \left[p_n^\mathrm{max}-p_0\left(A_{n(S_n)}\rho_n+B_{n(S_n)}\right)\right]}{\sum\limits^{N}_{k=1,k\neq n}U_{nk} \left[p_k^\mathrm{max}-p_0\left(A_{k(S_k)}\rho_k+B_{k(S_k)}\right)\right] + v_n}\Bigg),\tag{\ref{obj}}
    \end{multline}
\end{subequations}
which is the objective function of problem \eqref{s3}. Note that it is only related to the semantic compression ratio $\bm{\rho}$.

Thus, problem \eqref{s3} can be rewritten as
\begin{subequations}\label{s32}
    \begin{align}
        \max_{\bm{\rho}} \quad & f\left(\bm{\rho}\right), \tag{\ref{s32}}\\
        \textrm{s.t.} \quad & \rho_n\geq \frac{\left(p_n^\mathrm{max}/p_0\right)-B_{n(S_n)}}{A_{n(S_n)}},\forall n\in\mathcal{N}, \label{c321}\\
        & L_{n(S_n)}\leq\rho_n\leq L_{n(S_n-1)},\forall n\in\mathcal{N}. \label{c322}
    \end{align}
\end{subequations}

To begin, set the initial semantic compression ratio as
\begin{equation}\label{is}
    \bm{\rho}^{(0)}=\left[\rho_{1(S_1)},\rho_{2(S_2)},\cdots,\rho_{N(S_N)}\right].
\end{equation}
% where $L_{n(S_n-1)}$ is the upper bound of $\rho_n$.

Let $\bm{\rho}^{(t-1)}$ denote the semantic compression ratio obtained in the $(t-1)$-th iteration. Subsequently, we can calculate the gradient of the objective function $f\left(\bm{\rho}\right)$ at $\bm{\rho}^{(t-1)}$ according to the definition, i.e.,
\begin{align}\label{grad}
    \left[\nabla_{\bm{\rho}}f\left(\bm{\rho}^{(t-1)}\right)\right]_n &=\frac{\partial f\left(\bm{\rho}\right)}{\partial \left[\bm{\rho}\right]_n}\Bigg|_{\bm{\rho}=\bm{\rho}^{(t-1)}}\notag\\
    &=\lim_{\delta\to 0}\frac{f\left(\bm{\rho}^{(t-1)}+\delta\mathbf{o}^n_N\right)-f\left(\bm{\rho}^{(t-1)}\right)}{\delta},
\end{align}
where $\mathbf{o}^n_N$ is a Boolean vector of size $N\times 1$ with $[\mathbf{o}^n_N]_n=1$ and $[\mathbf{o}^n_N]_m=0,m\neq n$.

Then, we can update $\bm{\rho}^{(t)}$ in the $t$-th iteration towards the gradient ascent direction for a higher $f\left(\bm{\rho}\right)$. The update strategy can be written as
\begin{equation}\label{update}
    \bm{\rho}^{(t)}=\mathcal{B}\left\{\bm{\rho}^{(t-1)}+\tau^{(t)}\nabla_{\bm{\rho}}f\left(\bm{\rho}^{(t-1)}\right)\right\},
\end{equation}
where $\tau^{(t)}$ represents the step size in the $t$-th iteration, and $\mathcal{B}\left\{\bm{\rho}\right\}$ refers to a boundary function which ensures that the semantic compression ratio stays within the range determined by constraints \eqref{c321} and \eqref{c322}. Specifically, the boundary function $\mathcal{B}\left\{\bm{\rho}\right\}$ can be expressed as
\begin{equation}\label{boundary}
    \left[\mathcal{B}\left\{\bm{\rho}\right\}\right]_n=\left\{
    \begin{aligned}
        & [\bm{\rho}]_n^\mathrm{min}, &[\bm{\rho}]_n<[\bm{\rho}]_n^\mathrm{min},\\
        & [\bm{\rho}]_n, &[\bm{\rho}]_n^\mathrm{min}\leq [\bm{\rho}]_n\leq [\bm{\rho}]_n^\mathrm{max},\\
        & [\bm{\rho}]_n^\mathrm{max}, &[\bm{\rho}]_n>[\bm{\rho}]_n^\mathrm{max},
    \end{aligned}
    \right.
\end{equation}
where
\begin{equation}\label{min}
    [\bm{\rho}]_n^\mathrm{min}=\max\left\{\frac{\left(p_n^\mathrm{max}/p_0\right)-B_{n(S_n)}}{A_{n(S_n)}},L_{(S_n)}\right\},
\end{equation}
and
\begin{equation}\label{max}
    [\bm{\rho}]_n^\mathrm{max}=L_{n(S_n-1)}.
\end{equation}

Both the convergence rate and the ultimate outcome of the gradient ascent algorithm exhibit pronounced sensitivity to the chosen step size. Oversized step sizes may expedite convergence but risk non-convergence. Conversely, overly small step sizes encourage convergence with more iterations, although resulting in a more optimal solution. Consequently, this paper employs the backtracking linear search method to ascertain a judicious step size. Concretely, within $t$-th iteration, the step size initiates with a sizeable positive value, i.e., $\tau^{(t)}=\Bar{\tau}$, and diminishes gradually by repeating
\begin{equation}\label{tau}
    \tau^{(t)}\leftarrow\alpha\tau^{(t)},\alpha\in (0,1),
\end{equation}
until the Armijo–Goldstein condition is satisfied, expressed as
\begin{equation}\label{ag}
    f\left(\bm{\rho}^{(t)}\right)\geq f\left(\bm{\rho}^{(t-1)}\right)+\xi\tau^{(t)}\left\Vert \nabla_{\bm{\rho}}f\left(\bm{\rho}^{(t-1)}\right)\right\Vert^2_2,
\end{equation}
where $\xi\in(0,1)$ serves as a hyper-parameter regulating the step size magnitude.

The algorithm will terminate when the increase in $f\left(\bm{\rho}\right)$ between the two most recent iterations is less than a very small positive number, denote by $\epsilon$, or the algorithm reaches the maximum iteration limit of $T^\mathrm{max}$. Algorithm \ref{algo2} provides a summary of the gradient ascent algorithm.

\begin{algorithm}[ht]
\caption{Gradient Ascent Algorithm for Refined Search of Semantic Compression Ratio}\label{algo2}
\begin{algorithmic}[1]
    \STATE Initialize $\bm{\rho}^{(0)}$. Set iteration index $t=0$.
    \STATE Obtain $f\left(\bm{\rho}\right)$ according to \eqref{obj}.
    \REPEAT
        \STATE Calculate $\nabla_{\bm{\rho}}f\left(\bm{\rho}^{(t-1)}\right)$ according to \eqref{grad}.
        \STATE Initialize the step size $\tau^{(t)}=\Bar{\tau}$.
        \STATE Update $\bm{\rho}$ according to \eqref{update}.
        \REPEAT
            \STATE Diminish the step size according to \eqref{tau}.
            \STATE Update $\bm{\rho}$ according to \eqref{update}.
        \UNTIL{the Armijo–Goldstein condition \eqref{ag} is satisfied.}
        \STATE Set $t=t+1$.
    \UNTIL{$\left\vert f\left(\bm{\rho}^{(t)}\right)-f\left(\bm{\rho}^{(t-1)}\right)\right\vert<\epsilon$ or $t>T^\mathrm{max}$.}
    \STATE \textbf{Output}: Semantic compression ratio $\bm{\rho}$ for all users.
\end{algorithmic}
\end{algorithm}

In this stage, the non-smooth computation function $g_n(\rho_n)$ is degenerated to a linear function according to the Boolean matrix $\bm{\Theta}$ obtained in stage 2. Then, a gradient ascent algorithm is employed to tackle the non-convex problem \eqref{s3}. This stage outputs the refined semantic compression ratio $\bm{\rho}$ for all users.

\subsection{Algorithm Analysis}
The overall joint transmission and computation resource allocation algorithm for the multi-user PSC network is presented in Algorithm \ref{algo3}. Algorithm \ref{algo3} consists of three stages that are executed sequentially. Therefore, the overall complexity of Algorithm \ref{algo3} can be calculated as $\mathcal{O}(\text{Stage 1}) + \mathcal{O}(\text{Stage 2}) + \mathcal{O}(\text{Stage 3})$, where $\mathcal{O}(\text{Stage }i)$ denotes the computation complexity of stage $i$. The complexity of these three stages is analyzed as follows.

In stage 1, we derive the closed-form solution of the receive beamforming matrix $\mathbf{W}$ using the MMSE strategy. Therefore, the computation complexity of stage 1 lies in computing $\mathbf{W}$. To compute $\mathbf{W}$, we need to perform four matrix multiplications and one matrix inversion. Hence, the computation complexity of stage 1 can be expressed as $\mathcal{O}(MN^2 + M^2N + M^3)$.

In stage 2, we employ the AO method to obtain the Boolean matrix $\bm{\Theta}$. If we exhaustively search all possibilities of $\bm{\Theta}$, the computation complexity would be $\mathcal{O}(S^N)$, which is infeasible. Although the result obtained by the AO method may not be the globally optimal solution, it significantly reduces the complexity to $\mathcal{O}(I^\text{max}SN)$. In Algorithm \ref{algo1}, the computation complexity for calculating the objective value in line 6 is $\mathcal{O}(N^2)$. Therefore, the computation complexity of stage 2 is $\mathcal{O}(I^\text{max}SN^3)$.

In stage 3, we utilize the gradient ascent algorithm to search for the refined semantic compression ratio $\bm{\rho}$. In Algorithm \ref{algo2}, the computation complexity for calculating the gradient in line 4 is $\mathcal{O}(N^3)$. Let $B^\text{max}$ denote the maximum iterations of the backtracking linear search in lines 7 to 10 of Algorithm \ref{algo2}. Thus, the complexity of Algorithm \ref{algo2} is $\mathcal{O}(B^\text{max}N)$. Consequently, the computation complexity of stage 3 is $\mathcal{O}(T^\text{max}(N^3 + B^\text{max}N))$.

As a result, the total complexity of Algorithm \ref{algo3} can be expressed as $\mathcal{O}(MN^2 + M^2N + M^3 + I^\text{max}SN^3 + T^\text{max}(N^3 + B^\text{max}N))=\mathcal O(M^3+I^\text{max}SN^3)$ since $N\leq M$.

\begin{algorithm}[ht]
\caption{Joint Transmission and Computation Resource Allocation Algorithm for Multi-User PSC Network}\label{algo3}
\begin{algorithmic}[1]
    \STATE Initialize $\mathbf{W}$, $\mathbf{p}$, and $\bm{\rho}$.
    \STATE \textbf{Stage 1:}
        \STATE \quad Update the receive beamforming matrix $\mathbf{W}$ according to \eqref{mmse}.
    \STATE \textbf{Stage 2:}
        \STATE \quad Substitute the transmit power $\mathbf{p}$ with the semantic compression ratio $\bm{\rho}$ according to Theorem \ref{theorem1}.
        \STATE \quad Rewrite $g_n(\rho_n)$ according to \eqref{grho}.
        \STATE \quad Calculate $\rho_{ns}$ according to \eqref{midrho}.
        \STATE \quad Solve problem \eqref{s23} using Algorithm \ref{algo1}.
    \STATE \textbf{Stage 3:}
        \STATE \quad Update $g_n(\rho_n)$ according to \eqref{segment}.
        \STATE \quad Solve problem \eqref{s3} using Algorithm \ref{algo2}.
    \STATE \textbf{Output}: The optimized $\mathbf{W}$, $\mathbf{p}$ and $\bm{\rho}$.
\end{algorithmic}
\end{algorithm}

Since deducing the optimality of problem \eqref{pf} is challenging in theory, obtaining the globally optimal solution would generally require exponential computation complexity, which is unrealistic. Therefore, we propose Algorithm \ref{algo3} to provide a suboptimal solution for problem \eqref{pf} with polynomial computation complexity.

\section{Simulation Results}
In the simulations, the considered PSC network comprises 8 users, while the BS is equipped with 16 antennas. The multiple access channel matrix $\mathbf{H}$ is configured with a long-term channel power gain $\beta$ set to -90 dB, and the noise power is set to -10 dBm. Furthermore, we set the computation power coefficient to 1 and the maximum power limit to 30 dBm. For the semantic information extraction task based on the probability graph, we adopt the same parameters as in \cite{zhao2023joint}. A summary of the main system parameters is provided in Table~\ref{tb1}.% All simulation results are derived as the average of 100 runs, each initiated with distinct random seeds for the matrix $\mathbf{H}$.

\begin{table}[ht]
\centering
\caption{Main System Parameters}
\begin{tabular}{|c||c||c|}
    \toprule\hline
    \textbf{Parameter} & \textbf{Symbol}  & \textbf{Value} \\
    \hline
    Number of users & $N$ & 8 \\ \hline
    Number of antennas & $M$ & 16 \\ \hline
    Long-term channel power gain & $\beta$ & -90 dB \\ \hline
    Noise power & $\sigma^2$ & -10 dBm \\ \hline
    Computation power coefficient & $p_0$ & 1 \\ \hline
    Maximum power limit & $p_n^\mathrm{max}$ & 30 dBm \\ \hline
    Parameter in \eqref{grad} & $\delta$ & $10^{-9}$ \\ \hline
    Initial step size & $\Bar{\tau}$ & $10^{-3}$ \\ \hline
    Scaling factor in \eqref{tau}& $\alpha$ & 0.5 \\ \hline
    Hyper-parameter in \eqref{ag} & $\xi$ & 0.1 \\ \hline
    Threshold in Algorithm \ref{algo2} & $\epsilon$ & $10^{-6}$ \\ \hline
    Maximum iteration limit in Algorithm \ref{algo2} & $T^\mathrm{max}$ & 1000 \\
    \hline\bottomrule
\end{tabular}
\label{tb1}
\end{table}

The proposed multi-user PSC system, enhanced by the probability graph with joint transmission and computation optimization, is labeled as the `PSC' scheme. For comparisons, we incorporate several benchmark schemes as follows. 
\begin{itemize}
\item \textbf{`Non-semantic':} This benchmark scheme represents a conventional communication approach where the original data is directly transmitted without employing semantic compression. In this scheme, all users' power is allocated solely to transmission, without any optimization for joint transmission and computation.
\item \textbf{`PSC-S2':} This scheme is a simplified version of the `PSC' scheme, where the optimization process is performed only up to stage 2. The final result is the roughly estimated semantic compression ratio obtained from this stage.
\item \textbf{`PSC-ZF':} In this scheme, the ZF strategy is employed at stage 1. This means that the receive beamforming matrix $\mathbf{W}$ is calculated as $\mathbf{W}=\mathbf{H}(\mathbf{H}^\mathrm{H}\mathbf{H})^{-1}$. The remaining stages are the same with the `PSC' scheme.
% \item \textbf{`PSC-MRC':} This scheme employs the MRC strategy at stage 1. This means that the receive beamforming matrix $\mathbf{W}$ is calculated as $\mathbf{W}=\mathbf{H}$. The subsequent stages remain unchanged with the `PSC' scheme.
\end{itemize}

\begin{figure}[t]
\centering
\includegraphics[width=\linewidth]{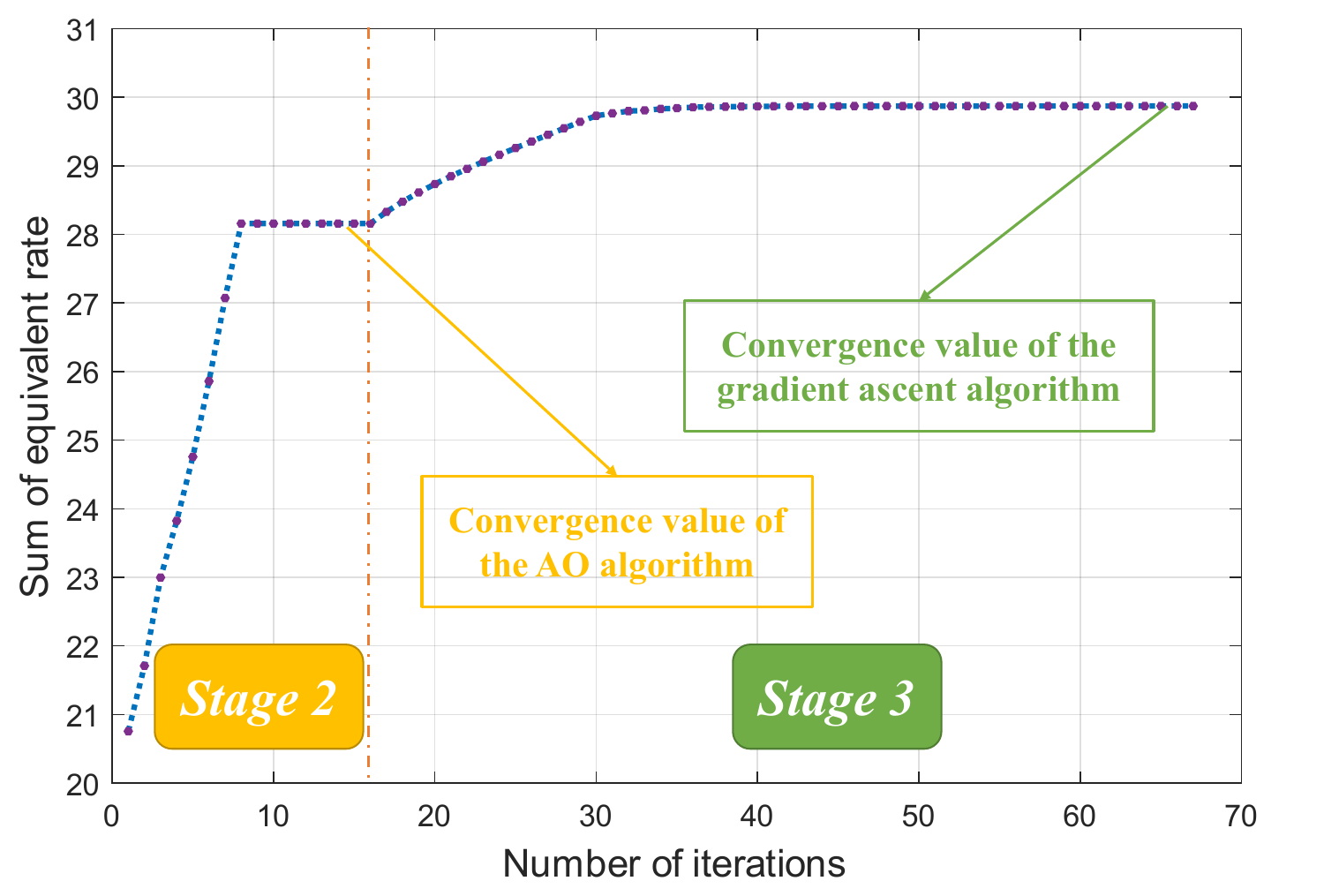}
\caption{Sum of equivalent rate vs. number of iterations.}
\label{fg:convergence}
\end{figure}

In Fig.~\ref{fg:convergence}, we assess the convergence of the proposed `PSC' scheme. Two convergent platforms are discernible: the first pertains to the AO algorithm, while the second corresponds to the gradient ascent algorithm. During stage~2, the objective value exhibits rapid ascent and subsequent convergence. This can be attributed to the fact that, in this stage, the AO algorithm addresses an integer programming problem with a discrete and relatively small variable space. Upon the convergence of the AO algorithm, the `PSC' scheme progresses to stage 3, wherein the gradient ascent algorithm is activated. In stage 3, the objective function converges to a value higher than that achieved in stage 2. This observation serves as validation for the effectiveness of the gradient ascent algorithm. Throughout the iterative process, the objective value steadily increases, eventually reaching a highly stable value. This outcome substantiates the efficacy of the comprehensive algorithm design.

\begin{figure}[t]
\centering
\includegraphics[width=\linewidth]{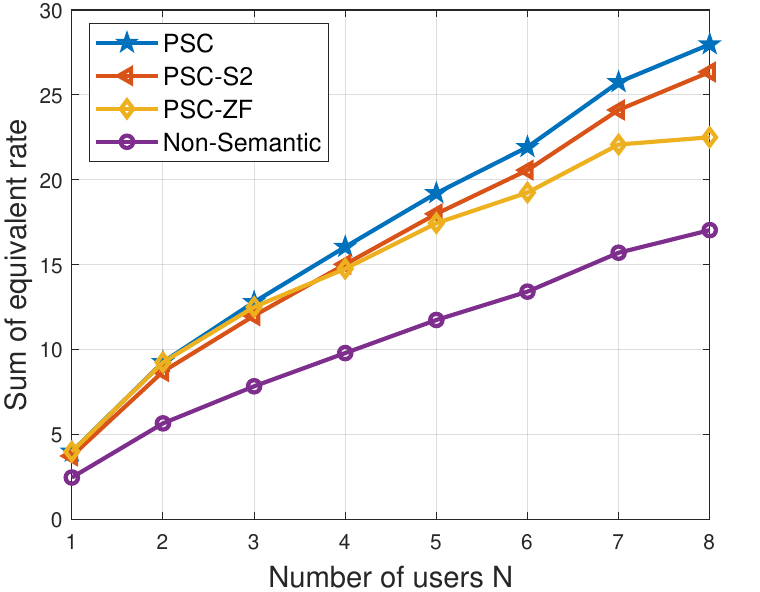}
\caption{Sum of equivalent rate vs. number of users.}
\label{fg:n}
\end{figure}

In Fig.~\ref{fg:n}, the correlation between the sum of equivalent rate and the number of users is depicted. The figure reveals a consistent increase in the sum of equivalent rate across all schemes as the number of users increases. However, it is observed that this increase does not follow a linear trend with a slope of one. Specifically, when $N=8$, the sum of equivalent rate is found to be less than twice as high as that when $N=4$ within the same scheme. This phenomenon is attributed to the emergence of inter-user interference at the receiver. Furthermore, the growth rate of the `PSC' scheme surpasses that of the `PSC-ZF' scheme, indicating that the MMSE strategy outperforms the ZF strategy in the examined scenario. It is important to emphasize that, consistently, the `PSC' scheme demonstrates the highest performance, while the sum rate of the `Non-semantic' scheme consistently remains the lowest.

\begin{figure}[t]
\centering
\includegraphics[width=\linewidth]{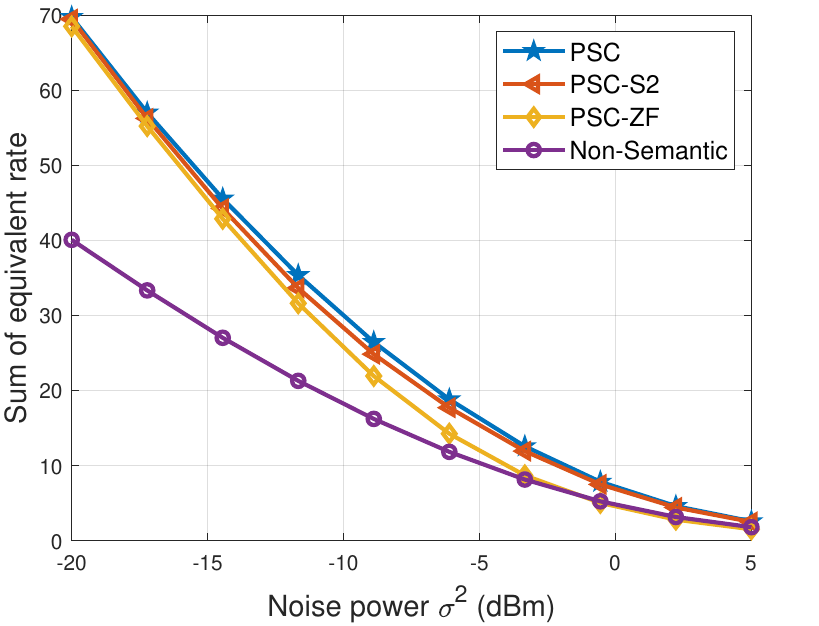}
\caption{Sum of equivalent rate vs. noise power.}
\label{fg:np}
\end{figure}

In Fig.~\ref{fg:np}, the variation of the sum of equivalent rate with changing noise power is illustrated. The figure highlights a consistent decrease in the sum of equivalent rate across all schemes as the noise power increases. When the noise power is small, the performance of the `PSC' scheme and the `PSC-ZF' scheme is comparable, suggesting that the ZF strategy is more effective in low-noise environments. It is important to note that, theoretically, when the noise power is zero, the formulas for both MMSE and ZF strategies yield identical results. However, in real-world scenarios, complete absence of noise is implausible. Consequently, the superiority of the MMSE strategy over the ZF strategy becomes evident as noise power increases. This is demonstrated in Fig.~\ref{fg:np}, where the `PSC' scheme consistently outperforms the `PSC-ZF' scheme across various noise power levels, affirming the general superiority of the MMSE strategy.

\begin{figure}[t]
\centering
\includegraphics[width=\linewidth]{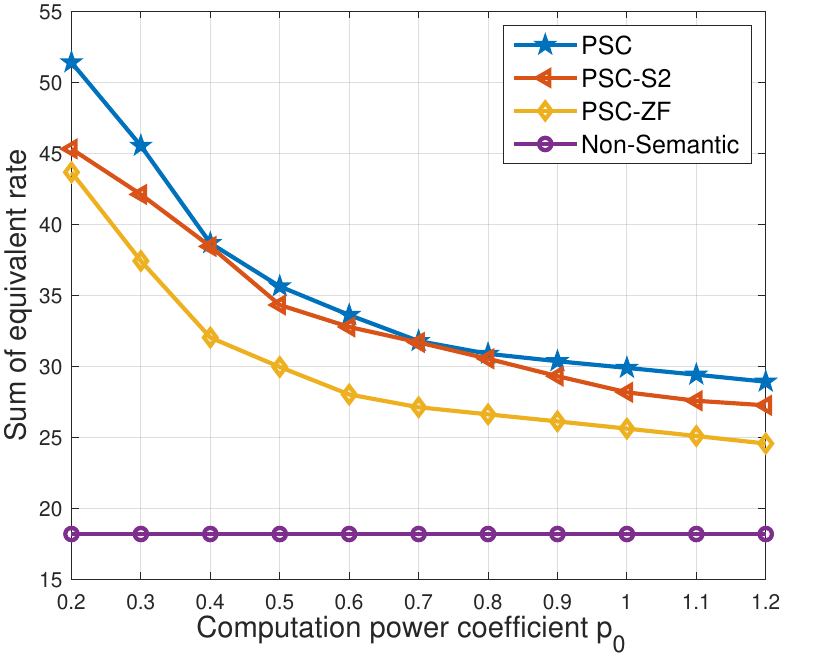}
\caption{Sum of equivalent rate vs. computation power coefficient.}
\label{fg:p0}
\end{figure}

In Fig.~\ref{fg:p0}, the relationship between the sum of equivalent rate and the computation power coefficient is depicted. Notably, the `Non-semantic' scheme maintains a constant sum of equivalent rate across different $p_0$ values due to its lack of utilization of semantic communication techniques, and consistently exhibiting the lowest performance among the considered schemes. As the computation power coefficient decreases, the sum of equivalent rate for the other three schemes increases. This trend is attributed to the enhanced efficiency in computation with lower $p_0$, facilitating a lower semantic compression ratio. Consequently, a higher sum of equivalent rate is achieved. It is found that the `PSC-S2' scheme exhibits variable proximity to the `PSC' scheme, illustrating a dynamic relationship. A small gap between the two indicates that the solution of the `PSC' scheme closely aligns with the midpoint solution of the `PSC-S2' scheme. Moreover, the sum of equivalent rate for the `PSC-S2' scheme demonstrates a segmented function concerning the computation power coefficient $p_0$. This behavior arises because the solution of the `PSC-S2' scheme jumps to the midpoint of another segment of the computation load function $g_n(\rho_n)$ only when $p_0$ changes significantly.
% Notably, the solution $\rho$ remains consistent within each segment.

\begin{figure}[t]
\centering
\includegraphics[width=\linewidth]{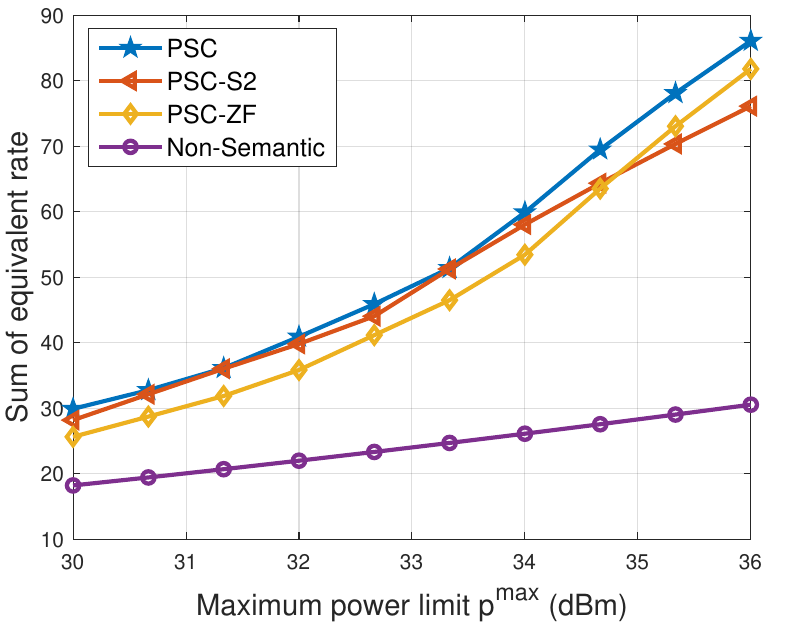}
\caption{Sum of equivalent rate vs. maximum power limit.}
\label{fg:pmax}
\end{figure}

In Fig.~\ref{fg:pmax}, the evolution of the sum of equivalent rate is traced across varying maximum power limits. A consistent upward trajectory is observed for all schemes as the maximum power limit increases. This behavior is a direct consequence of the positive correlation between augmented power levels and increased achievable rates for all users. Distinctly, in comparison to the `Non-semantic' scheme, the advantages of the `PSC' scheme become more pronounced with higher maximum power limits $p_n^\mathrm{max}$. This enhancement can be attributed to the `PSC' scheme's ability to allocate more power to semantic compression as the maximum power limit increases. The reduction in data size achieved through semantic compression significantly contributes to the overall sum of equivalent rate. Conversely, the `Non-semantic' scheme can only allocate all power to transmission, which does not contribute as significantly to the sum of equivalent rate. Consequently, the proposed `PSC' scheme exhibits substantial superiority when there is sufficient power.

\begin{figure}[t]
\centering
\includegraphics[width=\linewidth]{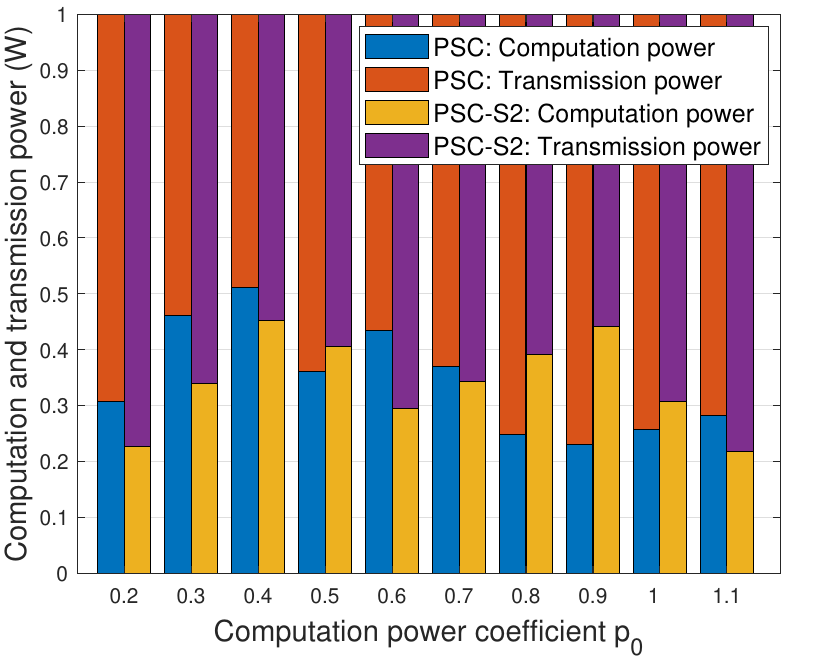}
\caption{The allocation of the computation power and transmission power with different computation power coefficient.}
\label{fg:power}
\end{figure}

To depict the allocation of computation power and transmission power within the considered network, Fig.~\ref{fg:power} illustrates the distribution in both the `PSC' and `PSC-S2' schemes across various computation power coefficients. It can be seen that the sum of computation power and transmission power consistently equals the predefined maximum power limit $p_n^\mathrm{max}$, set at 30 dBm. This figure reveals no discernible pattern in the variation of computation power with respect to $p_0$, and the computation power of the `PSC-S2' scheme fluctuates, at times surpassing and at other times falling below that of the `PSC' scheme. This variability underscores the inherent challenge in achieving a balance between transmission and computation within the considered PSC network.

\section{Conclusion}
This paper has introduced the PSC network, a novel paradigm where multiple users employ semantic information extraction techniques to compress extensive original data before transmission to a multi-antenna BS. Our model represents large-sized data through comprehensive knowledge graphs, utilizing a shared probability graph between users and the BS to facilitate efficient semantic compression. We formulated an optimization problem aimed at maximizing the sum of equivalent rate for all users, while considering total power constraints and semantic requirements. To tackle the non-convex and non-smooth nature of the optimization problem, we proposed a three-stage algorithm. This algorithm determines the receive beamforming matrix of the BS, transmit power, and semantic compression ratio for each user step by step. Numerical results underscore the effectiveness of our proposed scheme, emphasizing its ability to achieve a harmonious equilibrium between transmission and computation.

In future research, we plan to extend our exploration of resource management in the PSC network to diverse scenarios, such as unmanned aerial vehicle (UAV) networks, near-field communications, and other relevant domains. Additionally, considering the uniform computation power coefficient for every user in this study, it is worth investigating the performance of the PSC network among computing-heterogeneous devices. These avenues present interesting directions for future research in the PSC network.

\begin{appendices}
\section{Proof of Lemma \ref{lemma1}}
The received signals at the BS without beamforming can be expressed as
\begin{equation}\label{yh}
    \hat{\mathbf{y}}=\mathbf{H}\mathbf{x}+\mathbf{n},
\end{equation}
which means $\mathbf{y}=\mathbf{W}^\mathrm{H}\hat{\mathbf{y}}$ based on \eqref{rs} and \eqref{yh}.

The goal of the MMSE strategy is to minimize the mean square error (MSE) between the transmitted signals $\mathbf{x}$ and the received signals $\mathbf{y}$. The error between $\mathbf{x}$ and $\mathbf{y}$ is
\begin{equation}\label{e}
    \mathbf{e} = \mathbf{y}-\mathbf{x}=\mathbf{W}^\mathrm{H}\hat{\mathbf{y}}-\mathbf{x}.
\end{equation}
To minimize the MSE between $\mathbf{x}$ and $\mathbf{y}$, represented by $\mathbb{E}\left\{\mathbf{e}^\mathrm{H}\mathbf{e}\right\}$, where $\mathbb{E}\left\{\mathbf{\cdot}\right\}$ denotes the expected value of
a random variable, the following condition must be satisfied
\begin{equation}\label{cond}
    \mathbb{E}\left\{\mathbf{e}\hat{\mathbf{y}}^\mathrm{H}\right\}=\mathbf{0},
\end{equation}
which means there is no correlation between $\hat{\mathbf{y}}$ and $\mathbf{e}$. Condition \eqref{cond} is equivalent to the condition that minimizes $\mathbb{E}\left\{\mathbf{e}^\mathrm{H}\mathbf{e}\right\}$, because if the correlation between $\hat{\mathbf{y}}$ and $\mathbf{e}$ is non-zero, it can still be used to decrease $\mathbb{E}\left\{\mathbf{e}^\mathrm{H}\mathbf{e}\right\}$.

Substituting \eqref{e} into \eqref{cond}, we have
\begin{equation}\label{cond2}
    \mathbb{E}\left\{(\mathbf{W}^\mathrm{H}\hat{\mathbf{y}}-\mathbf{x})\hat{\mathbf{y}}^\mathrm{H}\right\}=\mathbf{0},
\end{equation}
which is equivalent to
\begin{equation}\label{cond3}
    \mathbf{W}^\mathrm{H}\mathbb{E}\left\{\hat{\mathbf{y}}\hat{\mathbf{y}}^\mathrm{H}\right\}-\mathbb{E}\left\{\mathbf{x}\hat{\mathbf{y}}^\mathrm{H}\right\}=\mathbf{0}.
\end{equation}
According to \eqref{cond3}, we have
\begin{equation}\label{cond4}
    \mathbf{W}^\mathrm{H}=\mathbb{E}\left\{\mathbf{x}\hat{\mathbf{y}}^\mathrm{H}\right\}\mathbb{E}\left\{\hat{\mathbf{y}}\hat{\mathbf{y}}^\mathrm{H}\right\}^{-1}.
\end{equation}

Let us deal with $\mathbb{E}\left\{\mathbf{x}\hat{\mathbf{y}}^\mathrm{H}\right\}$ first. Substituting \eqref{yh} into $\mathbb{E}\left\{\mathbf{x}\hat{\mathbf{y}}^\mathrm{H}\right\}$, we obtain
\begin{equation}\label{xy1}
    \mathbb{E}\left\{\mathbf{x}\hat{\mathbf{y}}^\mathrm{H}\right\}=\mathbb{E}\left\{\mathbf{x}\left(\mathbf{H}\mathbf{x}+\mathbf{n}\right)^\mathrm{H}\right\}=\mathbb{E}\left\{\mathbf{x}\mathbf{x}^\mathrm{H}\mathbf{H}^\mathrm{H}+\mathbf{x}\mathbf{n}^\mathrm{H}\right\}.
\end{equation}
Since there is no correlation between the transmitted signals $\mathbf{x}$ and the noise $\mathbf{n}$, i.e., $\mathbb{E}\left\{\mathbf{x}\mathbf{n}^\mathrm{H}\right\}=\mathbf{0}$, we have
\begin{equation}\label{xy3}
    \mathbb{E}\left\{\mathbf{x}\hat{\mathbf{y}}^\mathrm{H}\right\}=\mathbb{E}\left\{\mathbf{x}\mathbf{x}^\mathrm{H}\right\}\mathbf{H}^\mathrm{H}=\mathbf{P}\mathbf{H}^\mathrm{H}.
\end{equation}

Following the similar procedure, we can obtain
\begin{equation}\label{yy}
    \mathbb{E}\left\{\hat{\mathbf{y}}\hat{\mathbf{y}}^\mathrm{H}\right\}=\mathbf{H}\mathbb{E}\left\{\mathbf{x}\mathbf{x}^\mathrm{H}\right\}\mathbf{H}^\mathrm{H}+\mathbb{E}\left\{\mathbf{n}\mathbf{n}^\mathrm{H}\right\}=\mathbf{H}\mathbf{P}\mathbf{H}^\mathrm{H}+\sigma^2\mathbf{I}_M.
\end{equation}

Now, substituting \eqref{xy3} and \eqref{yy} into \eqref{cond4}, we have
\begin{equation}\label{wh}
    \mathbf{W}^\mathrm{H}=\mathbf{P}\mathbf{H}^\mathrm{H}\left(\mathbf{H}\mathbf{P}\mathbf{H}^\mathrm{H}+\sigma^2\mathbf{I}_M\right)^{-1},
\end{equation}
which is equivalent to
\begin{equation}\label{wp}
    \mathbf{W}=\left(\mathbf{H}\mathbf{P}\mathbf{H}^\mathrm{H}+\sigma^2\mathbf{I}_M\right)^{-1}\mathbf{H}\mathbf{P}.
\end{equation}
From \eqref{wp}, the obtained receive beamforming matrix is associated with the transmit power $\mathbf{P}$.
$\hfill\square$

\section{Proof of Theorem \ref{theorem1}}
Theorem \ref{theorem1} can be proved by the contradiction method. If there exists a user $n$ such that 
\begin{equation}\label{tp1}
    p_n^\mathrm{t}+g_n(\rho_n)p_0 < p_n^\mathrm{max}.
\end{equation}
Then, for user $n$, we can always decrease its semantic compression ratio $\rho_n$ due to \eqref{assumption} and constraint \eqref{c12}.

It is evident that the objective function of problem \eqref{s1} decreases monotonically for $\rho_n$, indicating that a lower semantic compression ratio $\rho_n$ produces a higher value of the objective function in problem \eqref{s1}. Therefore, when the objective function of problem \eqref{s1} reaches its maximum, the semantic compression ratio $\rho_n$ and transmit power $p_n^\mathrm{t}$ of each user must satisfy
\begin{equation}\label{tp2}
    p_n^\mathrm{t}+g_n(\rho_n)p_0=p_n^\mathrm{max},\forall n\in\mathcal{N}.
\end{equation}

Hence, Theorem \ref{theorem1} is proved.
$\hfill\square$
\end{appendices}

\bibliographystyle{IEEEtran}
\bibliography{main}

% Generated by IEEEtran.bst, version: 1.14 (2015/08/26)
\begin{thebibliography}{10}
\providecommand{\url}[1]{#1}
\csname url@samestyle\endcsname
\providecommand{\newblock}{\relax}
\providecommand{\bibinfo}[2]{#2}
\providecommand{\BIBentrySTDinterwordspacing}{\spaceskip=0pt\relax}
\providecommand{\BIBentryALTinterwordstretchfactor}{4}
\providecommand{\BIBentryALTinterwordspacing}{\spaceskip=\fontdimen2\font plus
\BIBentryALTinterwordstretchfactor\fontdimen3\font minus \fontdimen4\font\relax}
\providecommand{\BIBforeignlanguage}[2]{{%
\expandafter\ifx\csname l@#1\endcsname\relax
\typeout{** WARNING: IEEEtran.bst: No hyphenation pattern has been}%
\typeout{** loaded for the language `#1'. Using the pattern for}%
\typeout{** the default language instead.}%
\else
\language=\csname l@#1\endcsname
\fi
#2}}
\providecommand{\BIBdecl}{\relax}
\BIBdecl

\bibitem{10024766}
W.~Xu, Z.~Yang, D.~W.~K. Ng, M.~Levorato, Y.~C. Eldar, and M.~Debbah, ``Edge learning for {B5G} networks with distributed signal processing: Semantic communication, edge computing, and wireless sensing,'' \emph{IEEE J. Sel. Topics Signal Process.}, vol.~17, no.~1, pp. 9--39, Jan. 2023.

\bibitem{9771334}
K.~Lu, Q.~Zhou, R.~Li, Z.~Zhao, X.~Chen, J.~Wu, and H.~Zhang, ``Rethinking modern communication from semantic coding to semantic communication,'' \emph{IEEE Wireless Commun.}, vol.~30, no.~1, pp. 158--164, Feb. 2023.

\bibitem{9955525}
D.~Gündüz, Z.~Qin, I.~E. Aguerri, H.~S. Dhillon, Z.~Yang, A.~Yener, K.~K. Wong, and C.-B. Chae, ``Beyond transmitting bits: Context, semantics, and task-oriented communications,'' \emph{IEEE J. Sel. Areas Commun.}, vol.~41, no.~1, pp. 5--41, Nov. 2023.

\bibitem{10233741}
Z.~Zhao, Z.~Yang, Y.~Hu, L.~Lin, and Z.~Zhang, ``Semantic information extraction for text data with probability graph,'' in \emph{Proc. 2023 IEEE/CIC Int. Conf. Commun. China (ICCC Workshops)}, Aug. 2023.

\bibitem{chaccour2022less}
C.~Chaccour, W.~Saad, M.~Debbah, Z.~Han, and H.~V. Poor, ``Less data, more knowledge: Building next generation semantic communication networks,'' arXiv preprint arXiv:2211.14343, Nov. 2022.

\bibitem{10000901}
X.~Peng, Z.~Qin, D.~Huang, X.~Tao, J.~Lu, G.~Liu, and C.~Pan, ``A robust deep learning enabled semantic communication system for text,'' in \emph{Proc. 2022 IEEE Global Commun. Conf. (GLOBECOM)}, Dec. 2022, pp. 2704--2709.

\bibitem{9679803}
X.~Luo, H.-H. Chen, and Q.~Guo, ``Semantic communications: Overview, open issues, and future research directions,'' \emph{IEEE Wireless Commun.}, vol.~29, no.~1, pp. 210--219, Jan. 2022.

\bibitem{9763856}
L.~Yan, Z.~Qin, R.~Zhang, Y.~Li, and G.~Y. Li, ``Resource allocation for text semantic communications,'' \emph{IEEE Wireless Commun. Lett.}, vol.~11, no.~7, pp. 1394--1398, Jul. 2022.

\bibitem{9953095}
X.~Mu, Y.~Liu, L.~Guo, and N.~Al-Dhahir, ``Heterogeneous semantic and bit communications: A semi-noma scheme,'' \emph{IEEE J. Sel. Areas Commun.}, vol.~41, no.~1, pp. 155--169, Jan. 2023.

\bibitem{hu2023multiuser}
Z.~Hu, T.~Liu, C.~You, Z.~Yang, and M.~Chen, ``Multiuser resource allocation for semantic-relay-aided text transmissions,'' arXiv preprint arXiv:2311.06854, Nov. 2023.

\bibitem{9398576}
H.~Xie, Z.~Qin, G.~Y. Li, and B.-H. Juang, ``Deep learning enabled semantic communication systems,'' \emph{IEEE Trans. Signal Process.}, vol.~69, pp. 2663--2675, Apr. 2021.

\bibitem{10333452}
Z.~Zhao, Z.~Yang, Q.-V. Pham, Q.~Yang, and Z.~Zhang, ``Semantic communication with probability graph: A joint communication and computation design,'' in \emph{Proc. 2023 IEEE 98th Veh. Technol. Conf. (VTC2023-Fall)}, Oct. 2023.

\bibitem{yang2023secure}
Z.~Yang, M.~Chen, G.~Li, Y.~Yang, and Z.~Zhang, ``Secure semantic communications: Fundamentals and challenges,'' arXiv preprint arXiv:2301.01421, Jan. 2023.

\bibitem{zhao2023joint}
Z.~Zhao, Z.~Yang, X.~Gan, Q.-V. Pham, C.~Huang, W.~Xu, and Z.~Zhang, ``A joint communication and computation design for semantic wireless communication with probability graph,'' arXiv preprint arXiv:2312.13975, Dec. 2023.

\bibitem{yang2023energy}
Z.~Yang, M.~Chen, Z.~Zhang, and C.~Huang, ``Energy efficient semantic communication over wireless networks with rate splitting,'' \emph{IEEE J. Sel. Areas Commun.}, vol.~41, no.~5, pp. 1484--1495, May 2023.

\bibitem{10001594}
L.~Yan, Z.~Qin, R.~Zhang, Y.~Li, and G.~Ye~Li, ``{QoE-Aware} resource allocation for semantic communication networks,'' in \emph{Proc. 2022 IEEE Global Commun. Conf. (GLOBECOM)}, Dec. 2022, pp. 3272--3277.

\bibitem{10012845}
Z.~Yang, M.~Chen, Z.~Zhang, C.~Huang, and Q.~Yang, ``Performance optimization of energy efficient semantic communications over wireless networks,'' in \emph{Proc. 2022 IEEE 96th Veh. Technol. Conf. (VTC2022-Fall)}, Sep. 2022.

\bibitem{cang2023resource}
Y.~Cang, M.~Chen, Z.~Yang, Y.~Hu, Y.~Wang, Z.~Zhang, and K.-K. Wong, ``Resource allocation for semantic-aware mobile edge computing systems,'' arXiv preprint arXiv:2309.11736, Sep. 2023.

\bibitem{10183794}
Z.~Qin, F.~Gao, B.~Lin, X.~Tao, G.~Liu, and C.~Pan, ``A generalized semantic communication system: From sources to channels,'' \emph{IEEE Wireless Commun.}, vol.~30, no.~3, pp. 18--26, Jun. 2023.

\bibitem{9953076}
D.~Huang, F.~Gao, X.~Tao, Q.~Du, and J.~Lu, ``Toward semantic communications: Deep learning-based image semantic coding,'' \emph{IEEE J. Sel. Areas Commun.}, vol.~41, no.~1, pp. 55--71, Jan. 2023.

\bibitem{9953316}
T.~Han, Q.~Yang, Z.~Shi, S.~He, and Z.~Zhang, ``Semantic-preserved communication system for highly efficient speech transmission,'' \emph{IEEE J. Sel. Areas Commun.}, vol.~41, no.~1, pp. 245--259, Jan. 2023.

\bibitem{9450827}
Z.~Weng and Z.~Qin, ``Semantic communication systems for speech transmission,'' \emph{IEEE J. Sel. Areas Commun.}, vol.~39, no.~8, pp. 2434--2444, Aug. 2021.

\bibitem{10061867}
L.~Hu, Y.~Li, H.~Zhang, L.~Yuan, F.~Zhou, and Q.~Wu, ``Robust semantic communication driven by knowledge graph,'' in \emph{Proc. 2022 9th Int. Conf. Internet Things: Syst., Mgt. Sec. (IOTSMS)}, Nov. 2022.

\bibitem{9685056}
Y.~Wang, M.~Chen, W.~Saad, T.~Luo, S.~Cui, and H.~V. Poor, ``Performance optimization for semantic communications: An attention-based learning approach,'' in \emph{Proc. 2021 IEEE Global Commun. Conf. (GLOBECOM)}, Dec. 2021.

\bibitem{9357868}
M.~Gaur, K.~Faldu, and A.~Sheth, ``Semantics of the black-box: Can knowledge graphs help make deep learning systems more interpretable and explainable?'' \emph{IEEE Internet Computing}, vol.~25, no.~1, pp. 51--59, Feb. 2021.

\bibitem{8461983}
N.~Farsad, M.~Rao, and A.~Goldsmith, ``Deep learning for joint source-channel coding of text,'' in \emph{Proc. 2018 IEEE Int. Conf. Acoust. Speech. Signal. Process. (ICASSP)}, Sept. 2018, pp. 2326--2330.

\bibitem{9834044}
S.~Yao, K.~Niu, S.~Wang, and J.~Dai, ``Semantic coding for text transmission: An iterative design,'' \emph{IEEE Trans. Cogn. Commun. Netw.}, vol.~8, no.~4, pp. 1594--1603, Jul. 2022.

\bibitem{10118916}
C.~Liu, C.~Guo, S.~Wang, Y.~Li, and D.~Hu, ``Task-oriented semantic communication based on semantic triplets,'' in \emph{Proc. 2023 IEEE Wireless Commun. Netw. Conf. (WCNC)}, Mar. 2023.

\bibitem{9838470}
F.~Zhou, Y.~Li, X.~Zhang, Q.~Wu, X.~Lei, and R.~Q. Hu, ``Cognitive semantic communication systems driven by knowledge graph,'' in \emph{Proc. 2022 IEEE Int. Conf. Commun. (ICC)}, May. 2022, pp. 4860--4865.

\bibitem{9832831}
Y.~Wang, M.~Chen, T.~Luo, W.~Saad, D.~Niyato, H.~V. Poor, and S.~Cui, ``Performance optimization for semantic communications: An attention-based reinforcement learning approach,'' \emph{IEEE J. Sel. Areas Commun.}, vol.~40, no.~9, pp. 2598--2613, Jul. 2022.

\bibitem{6861946}
M.~Erol-Kantarci and H.~T. Mouftah, ``Energy-efficient information and communication infrastructures in the smart grid: A survey on interactions and open issues,'' \emph{IEEE Commun. Surveys Tuts.}, vol.~17, no.~1, pp. 179--197, Jul. 2014.

\bibitem{9039685}
J.~Li, A.~Sun, J.~Han, and C.~Li, ``A survey on deep learning for named entity recognition,'' \emph{IEEE Trans. Knowl. Data Eng.}, vol.~34, no.~1, pp. 50--70, Jan. 2022.

\bibitem{9446853}
Y.~Hu, H.~Shen, W.~Liu, F.~Min, X.~Qiao, and K.~Jin, ``A graph convolutional network with multiple dependency representations for relation extraction,'' \emph{IEEE Access}, vol.~9, pp. 81\,575--81\,587, Jun. 2021.

\end{thebibliography}

% \begin{IEEEbiography}[{\includegraphics[width=1in,height=1.25in,clip,keepaspectratio]{fig1}}]{Michael Shell}
% Use $\backslash${\tt{begin\{IEEEbiography\}}} and then for the 1st argument use $\backslash${\tt{includegraphics}} to declare and link the author photo.
% Use the author name as the 3rd argument followed by the biography text.
% \end{IEEEbiography}

\end{document}